\documentclass[acmsmall]{acmart}

\usepackage{tikz}
\usepackage{amsmath}

\usepackage{listings}
\usepackage{algorithm}
\usepackage[rightComments=true, indLines=true]{algpseudocodex}
\usepackage{graphicx}
\usepackage{textcomp}
\usepackage{xcolor}
\usepackage{subcaption}
\usepackage[font=small, skip=4pt]{caption}
\usepackage{circledsteps}
\usepackage{multirow}
\usepackage{balance}
\usepackage[subtle, lists=tight, mathdisplays=tight]{savetrees}
\usepackage{pdfpages}
\usepackage{circledsteps}

\renewcommand\footnotetextcopyrightpermission[1]{} 

\AtBeginDocument{%
  }

\setcopyright{cc}
\setcctype{by}
\acmJournal{PACMMOD}
\acmYear{2026} \acmVolume{4} \acmNumber{2 (SIGMOD)} \acmArticle{175}
\acmMonth{6} \acmDOI{10.1145/3802052}

\newcommand{\scmacro}{{\itshape Minerva}}
\algrenewcommand\textproc{}
\begin{document}

\title{Epoch-based Optimistic Concurrency Control \\in Geo-replicated Databases}

\renewcommand{\shorttitle}{Epoch-based Optimistic Concurrency Control in Geo-replicated Databases}

\author{Yunhao Mao}
\email{yunhao.mao@mail.utoronto.ca}
\affiliation{%
  \institution{University of Toronto}
  \city{Toronto}
  \country{Canada}
}

\author{Harunari Takata}
\email{tkt0430@keio.jp}
\affiliation{%
  \institution{Keio University}
  \country{Japan}
}

\author{Michail Bachras}
\email{michalis.bachras@mail.utoronto.ca}
\affiliation{%
  \institution{University of Toronto}
  \city{Toronto}
  \country{Canada}
}

\author{Yuqiu Zhang}
\email{quincy.zhang@mail.utoronto.ca}
\affiliation{%
  \institution{University of Toronto}
  \city{Toronto}
  \country{Canada}
}

\author{Shiquan Zhang}
\email{shiquan.zhang@mail.utoronto.ca}
\affiliation{%
  \institution{University of Toronto}
  \city{Toronto}
  \country{Canada}
}

\author{Gengrui Zhang}
\email{gengrui.zhang@concordia.ca}
\affiliation{%
  \institution{Concordia University}
  \city{Montreal}
  \country{Canada}
}

\author{Hans-Arno Jacobsen}
\email{jacobsen@eecg.toronto.edu}
\affiliation{%
  \institution{University of Toronto}
  \city{Toronto}
  \country{Canada}
}

\renewcommand{\shortauthors}{Y. Mao et al.}

\begin{abstract}
Achieving high-performance transaction processing in geo-replicated OLTP databases is challenging due to the extensive over-coordination in distributed atomic commitment, concurrency control, and fault-tolerant replication protocols. To address this issue, we introduce \scmacro{}, a unified distributed concurrency control protocol designed for highly scalable \textit{multi-leader} replication. \scmacro{} employs a novel epoch-based asynchronous replication protocol that decouples data propagation from the commitment process, enabling continuous transaction replication. Optimistic concurrency control is used to allow replicas to execute transactions concurrently and to commit without coordination. For conflict detection during validation, we construct a conflict graph and use a maximum weight independent set search algorithm to select the optimal subset of non-conflicting transactions for commitment, minimizing the number of invalid transactions. Finally, we \textit{deterministically re-execute} conflicting transactions, ensuring serializability while eliminating aborts. Our evaluation demonstrates that \scmacro{} outperforms state-of-the-art replicated databases, achieving over $3\times$ higher throughput in scalability experiments and $2.8\times$ higher throughput in a high-latency network simulation with the TPC-C benchmark.

\end{abstract}





\settopmatter{printfolios=true}
\maketitle

\section{Introduction}

Geo-replicated distributed databases are widely used in modern large-scale online applications. In this architecture, application data is replicated to multiple data centers across a continent or even globally~\cite{sovran_transactional_2011, corbettSpannerGoogleGlobally2013, s3}, providing redundancy that significantly expands fault-tolerance capabilities, from failures of a single server to failures across data centers or even the Internet infrastructure of an entire region\footnote{For example, Bell, one of Canada's Internet infrastructure providers, caused a province-wide outage in 2025 due to misconfiguration~\cite{belloutage}.}~\cite{fboutage2021, awsoutage2025}. 

Prominent distributed OLTP databases, such as FoundationDB~\cite{zhou2021foundationdb}, TiDB~\cite{huang2020tidb}, or traditional relational databases with replication enabled, such as PostgreSQL~\cite{momjian2001postgresql}, and MySQL~\cite{widenius2002mysql}, usually replicate with a \textit{primary-backup} (single-leader) architecture. In this paradigm, write operations are processed exclusively by a single replica (the \textit{primary}), while secondaries serve as failover backups or read-only sources. Doing so allows simpler concurrency control and replication protocols but presents several limitations: (1) during network partitions, availability is compromised if clients cannot reach the designated primary; (2) the single primary often becomes a performance bottleneck; and (3) primary failure requires a leader change, which is time-consuming and complex with durability guarantees~\cite{kleppmannDesigningDataIntensiveApplications2016}.

To fully leverage the availability advantages of geo-replication, \textit{multi-leader} (active-active, multi-master) replication is a desired approach, where all replicas are able to process both read and write requests, and propagate changes to other replicas. Doing so ensures service availability during failures or network partitions, as the clients can be routed to any surviving replica. Furthermore, clients can be distributed among replicas for load balancing and low-latency access. 

However, achieving high performance in multi-leader transactional geo-replicated databases remains a significant challenge due to the extensive coordination overhead required to ensure ACID properties across replicas~\cite{bailisHAT2013, bailis_coordination_2014, zhang_tapir_2018, maiyyaCNC2019, mu_consolidating_2016, thomsonCalvinFastDistributed2012, fanOceanVistaGossipbased2019}. Specifically, the system must enforce strong consistency and coordinate global concurrency control to prevent conflicting updates. These synchronization protocols often require multi-round communication, where each round trip incurs significant wide-area network (WAN) latency, ranging from 80 to 100 ms across a continent (e.g., within North America) to over 200 ms for intercontinental connections (e.g., the US to Australia)~\cite{bailisHAT2013}.

Consequently, multi-leader replication is mostly seen in eventually consistent data stores, such as DynamoDB~\cite{awsdynamo2012}, Cassandra~\cite{lakshmanCassandraDecentralizedStructured2010a}, ScyllaDB~\cite{SycallaDB},  Redis~\cite{ActiveActiveGeodistributedRedis}, Riak~\cite{Riak}, and Anna~\cite{wuAnnaKVSAny2021b}, because weaker consistency requires less coordination~\cite{bailisHAT2013}. However, the relaxation of consistency guarantees means that full ACID transactions are no longer supported. As a result, developers must carefully reason about application correctness in the presence of stale or inconsistent reads.

To address these challenges, we present \scmacro{}, a novel distributed concurrency control protocol designed to deliver \textit{high-performance} and \textit{scalable} transactions that guarantee \textit{serializable isolation} and \textit{crash-fault tolerance} within a \textit{multi-leader} architecture. \scmacro{} is a unified protocol that combines fault-tolerant replication and concurrency control. The strong consistency and multi-leader replication also guarantee atomic commitment of transactions without requiring 2-Phase Commit (2PC). 

The key design features of \scmacro{} are as follows. We first introduce a novel batch-based asynchronous replication protocol designed to achieve scalability under high network latency. In our approach, replicas continuously deliver transactions in batches and establish a write quorum. Decoupled from this data path, a coordinator periodically invokes a Raft-based~\cite{diegoRaft2014} consensus process to confirm \textit{the indices} of the batches to commit. The separation of consensus agreement from bulk data replication significantly enhances system throughput.

To both maximize concurrency and eliminate aborts, \scmacro{} employs a two-phased, coordination-free transaction execution model that combines \textit{optimistic concurrency control (OCC)} and \textit{deterministic execution}. In the first phase, transactions are executed optimistically on local database snapshots, allowing high concurrency without locking. During replication, both transaction inputs (e.g., query parameters) and the write-sets are replicated. Upon commit, replicas independently gather the transactions in the batches, deterministically validate the transactions, then \textit{re-execute} the conflicting (read-write and write-write conflicts) transactions with the same input rather than aborting them, ensuring serializability. This approach is inspired by \textit{deterministic concurrency control}~\cite{thomsonddbv2010,thomsonCalvinFastDistributed2012, abadiOverviewDeterministicDatabase2018, luAriaFastPractical2020}. Every commit constitutes an \textit{epoch} that brings the database state to a new snapshot upon each epoch transition. Thus, the only coordination is the light-weight agreement on batch indices that is run asynchronously from the replication, transaction execution, and validation.

Finally, to minimize the number of re-executed transactions, we construct a \textit{conflict graph}~\cite{bernsteinconcurrency1981} representing the concurrent conflicting transactions and then search for a maximum weight independent set (MWIS) in the graph to get the maximum subset of non-conflicting transactions for commitment. 

We summarize our contributions as follows:
\begin{itemize}
    \item We present \scmacro{}, a distributed concurrency control protocol that enables scalable, fault-tolerant multi-leader transactions.
    \item We propose a novel epoch-based asynchronous replication protocol that decouples data propagation from consensus, achieving high throughput even in high-latency environments.
    \item We design a two-phase coordination-free transaction execution model and an optimal conflict resolution strategy using MWIS to maximize concurrency and minimize aborts while ensuring serializability.
    \item We implement \scmacro{} in a transactional key-value database and conduct thorough evaluations comparing it against state-of-the-art distributed databases. Our results demonstrate that \scmacro{} achieves comparable throughput in a low-latency (local) setting and significantly outperforms other baselines in high-latency and highly scaled scenarios.
\end{itemize}

The remainder of this paper is organized as follows: Section~\ref{sec:bg} provides the necessary background. Section~\ref{sec:design} details \scmacro{}'s algorithm. Section~\ref{sec:impl} describes the implementation of our OLTP database. Section~\ref{sec:eval} presents the experimental evaluation, and Section~\ref{sec:related} reviews related work.

\section{Background and Motivation}
\label{sec:bg}

\subsection{Challenges in Multi-leader Replication}

In a traditional primary-backup architecture, concurrency control, atomic commitment, and replication are typically handled as orthogonal components because writes to a single data item are handled by only one replica (the primary or the leader)~\cite{zhang_tapir_2018}. Concurrency control is only enforced on the primary, replication only ensures consistency between the primary and its followers, and atomic commit only involves the nodes that are participating in the transaction. As a result, coordination overhead remains confined to a small subset of nodes~\cite{kleppmannDesigningDataIntensiveApplications2016}.

In contrast, multi-leader replication allows all replicas to write. Therefore, concurrency control, atomic commitment, and replication must be coordinated globally across all replicas to guarantee serializability. The main challenge lies in the extensive multi-round, often repeated communication required by these protocols. This overhead becomes prohibitively expensive in geo-distributed settings due to high WAN latency~\cite{zhang_tapir_2018,maiyyaCNC2019,mu_consolidating_2016,thomsonCalvinFastDistributed2012,fanOceanVistaGossipbased2019,bailisHAT2013}. Figure~\ref{fig:mmvsms} illustrates the differences between the primary-backup and the multi-leader architecture.

\begin{figure}
    \centerline{\includegraphics[scale=0.7] {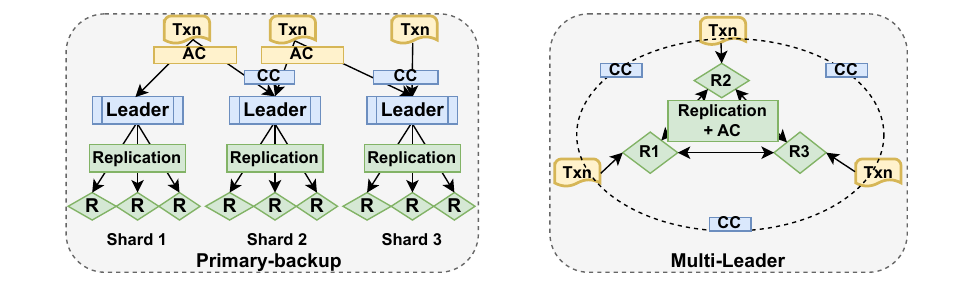}}
    \caption{Primary-backup vs. multi-leader replication. In multi-leader, the layers for atomic commit (\textbf{AC}), concurrency control (\textbf{CC}), and replication among replicas (\textbf{R}) are no longer orthogonal to each other as in primary-backup, leading to redundant communication.}
    \label{fig:mmvsms}
\end{figure}

Next, we detail the requirements for each of these three aspects. First, replicas must maintain strong consistency to prevent anomalies such as \textit{write skew} when read-write transactions are executed on multiple replicas~\cite{bailisHAT2013}. Consider a scenario where transaction $T_1$ writes $X=2$ and commits on one replica, while $T_2$ reads a stale value $X=1$ from another replica and writes $Y=X+1$. If $T_2$ commits before receiving $T_1$'s update, $T_1$ and $T_2$ are not serializable because no serial execution order can produce this result. Preventing such anomalies requires strong consistency, which typically requires synchronous replication via \textit{consensus} protocols~\cite{gilbert2002CAP,lamportPaxosMadeSimple2001,diegoRaft2014}.

Second, concurrency control must be coordinated globally. With a single primary, concurrency control is localized at the primary replica. However, with multi-leader, the system must detect concurrent reads and writes to the same data items when transactions are executed concurrently on different replicas to prevent \textit{conflicting updates}. For example, standard Two-Phase Locking (2PL) requires a transaction to acquire locks on data items across multiple replicas, and thus needs to coordinate~\cite{bernsteinconcurrency1981, cowlingGranolaLowOverheadDistributed2012, balakrishnanTangoDistributedData2013}. Similarly, OCC executes transactions concurrently but requires a synchronized validation phase to detect conflicts and agree on the valid transactions~\cite{kungOptimisticMethodsConcurrency1981, ceriOCC1982, bernsteinconcurrency1981}. Thus, both methods rely on consensus among replicas to guarantee serializability~\cite{thomasianDistributedOptimisticConcurrency1998, ozsuPrinciplesDistributedDatabase2020}.

Finally, atomic commitment protocols, such as 2PC, must be coordinated to ensure all replicas agree on the final outcome (commit or abort) of a transaction. This process again involves coordination overhead comparable to consensus~\cite{grayConsensusTransactionCommit2006}.

\subsection{Toward Unified Protocols}
To mitigate the overheads, prior research has proposed various optimizations designed to reduce communication overhead in distributed protocols, including the sharded primary-back up pattern. These strategies include leveraging inconsistent replication and delegating ordering to an asynchronous layer (e.g., TAPIR~\cite{zhang_tapir_2018}), decoupling replication from atomic commitment (e.g., Ocean Vista~\cite{fanOceanVistaGossipbased2019}), and consolidating distinct protocol phases to minimize communication rounds (e.g., Janus~\cite{mu_consolidating_2016} and the C\&C framework~\cite{maiyyaCNC2019}).

More recently, \textit{deterministic concurrency control} systems, such as Calvin~\cite{thomsonCalvinFastDistributed2012} and Aria~\cite{luAriaFastPractical2020}, have adopted the strategy of replicating transaction \textit{inputs} rather than the resulting states. In this paradigm, replicas execute transactions independently but adhere to a deterministic order. This approach eliminates the need for distributed atomic commitment and concurrency control, as the pre-established ordering guarantees serializability and consistency. By amortizing the coordination overhead across batches of transactions, these systems achieve high scalability. However, this comes at the cost of the inability to support interactive transactions and the complexity in agreeing on the commit point for each batch of transactions.

\subsection{Asynchronous Consensus}
Another critical challenge is to design a scalable, fault-tolerant replication mechanism. Traditional consensus protocols, such as Paxos~\cite{lamportPaxosMadeSimple2001} and Raft~\cite{diegoRaft2014} are synchronous in nature: a value is proposed, agreement is reached, and the value is persisted before the system state advances, and only then can the system advance to the next value. In high-concurrency environments, this causes significant blocking, where the agreement latency for one value delays the replication of all subsequent values~\cite{giridharanautobahn2024}.

Recently, some Byzantine Fault Tolerant (BFT) protocols, such as Narwhal \& Tusk~\cite{danezis2022narwhal} and Autobahn~\cite{giridharanautobahn2024}, address this limitation by decoupling message dissemination from the agreement process. The core idea is to propagate values continuously without blocking for immediate consensus, while a separate, asynchronous process periodically establishes the global order for batches of values using causal dependencies. These approaches offer high throughput consensus, but they are primarily designed for BFT systems. We adapt these BFT-oriented algorithms to be more efficient in a crash fault-tolerant scenario and integrate them into \scmacro{}'s transaction replication mechanism.

\paragraph{Summary}
The approaches discussed above address specific challenges in building scalable, geo-replicated multi-leader databases. However, they target fundamentally different architectures and prioritize distinct trade-offs. Inspired by these predecessors, we aim to create a unified protocol that integrates concurrency control, strongly consistent replication, and atomic commitment. \scmacro{} adopts the "execute-first" pattern of OCC to maximize horizontal scalability and remove coordination during the initial execution. The deterministic validation and commitment phase also allows for a coordination-free commit process. Finally, we leave the only coordination to a high-throughput asynchronous consensus protocol that decouples data propagation from agreement. These design choices collectively enable \scmacro{} to achieve high performance and scalability in geo-distributed multi-leader environments.

\section{\scmacro{}'s Design}
\label{sec:design}
To achieve our performance goals, \scmacro{}'s key design highlights are as follows:

\begin{itemize}
    \item A two-phased concurrency control mechanism that combines optimistic local execution with deterministic re-execution, thereby maximizing concurrency while minimizing abort costs under contention while being coordination-free.
    \item A novel high-throughput, fault-tolerant replication protocol that decouples data propagation from the agreement process, reducing communication overhead in geo-replicated environments.
    \item A graph-based conflict detection algorithm that minimizes transaction aborts during validation by modeling the optimal commit set as the MWIS problem. 
\end{itemize}

Next, we present the system model and assumptions, followed by a detailed description of each phase of the \scmacro{} protocol.

\subsection{System Architecture and Assumptions}
\scmacro{} is designed with geo-distributed environments in mind where each replica is deployed in a distinct data center, forming a replication cluster that maintains a copy of the database, as illustrated in Figure~\ref{fig:arch}. Clients may submit transactions to any replica, which then executes the request and propagates the changes to its peers. Within a single replica, transaction execution threads and storage nodes can be scaled horizontally with more machines. For simplicity, the remainder of this paper treats each replica as a single logical unit and assumes the replication cluster conducts full database replication~\footnote{If one replication cluster only contains a partition of the database and there are multiple clusters, 2PC is required to coordinate commit for cross-partition transactions. This can be achieved using the client as the 2PC coordinator, but it is out of the scope of this work.}.

\begin{figure}
    \centering
    \includegraphics[scale=0.65]{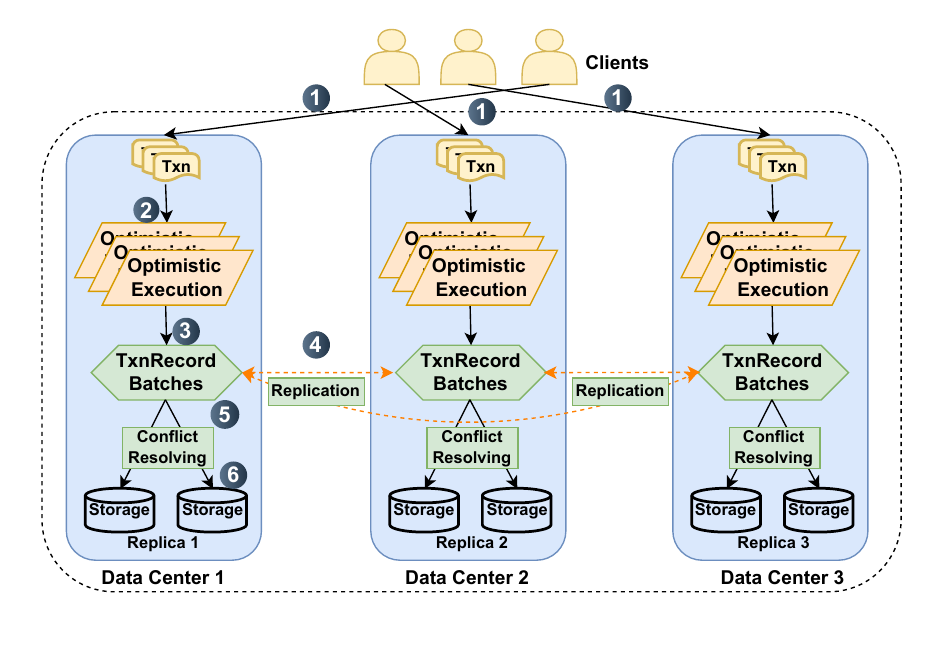}
    \caption{Architecture of \scmacro{}, illustrating the transaction workflow across three replicas. \Circled{1} Clients submit transactions to any replica. \Circled{2} Each replica independently and optimistically executes transactions. \Circled{3} Transaction inputs, read-sets, and write-sets are encapsulated in \textit{TransactionRecords}. \Circled{4} These records are batched and asynchronously propagated to peer replicas. \Circled{5} Periodically, a consensus process is used to decide on the batches to commit, followed by independent conflict resolution and deterministic re-execution. \Circled{6} Committed results are persisted and returned to clients.}
    \label{fig:arch}
\end{figure}

We assume a crash failure model where the replication cluster tolerates up to $f$ failures among $n \ge 2f + 1$ replicas in the cluster. Replicas are considered either correct or crashed. We treat failures at the replica granularity; if any internal component fails, the entire replica is deemed faulty. The communication channel is assumed to be asynchronous but reliable (like TCP): messages may be delayed indefinitely (in case of network partition) but cannot be lost or reordered. To ensure liveness, we assume that during any network partition, a quorum of at least $n-f$ replicas remains connected\footnote{We do not address Byzantine faults in this work. However, our design can be extended to support Byzantine Fault Tolerance (BFT) by integrating asynchronous BFT consensus algorithms, such as Autobahn~\cite{giridharanautobahn2024} or Narwhal~\cite{danezis2022narwhal}.}.

We define a transaction as a sequence of read, write, and computational operations. A transaction's \textit{input} consists of the client request (e.g., SQL query strings or stored procedure parameters), and its \textit{result} is defined as the write-set generated upon execution. We consider one-shot transactions but will discuss the extension to interactive transactions in \S~\ref{sec:ob}. Transactions go through four phases in \scmacro{}:

\begin{enumerate}
    \item \textbf{Snapshot-based Optimistic Local Execution} (\S~\ref{sec:local}): Clients submit requests to any replica based on locality or load. The receiving replica executes the transaction speculatively using a snapshot-based OCC protocol. Operations are performed against a \textit{consistent snapshot}, derived from the database state at the end of the previous epoch, with all write results buffered in a \textit{temporary state}.

    \item \textbf{Asynchronous Batched Replication} (\S~\ref{sec:replication}): 
    Locally executed transactions and their results are placed in a \textit{batch} and broadcast to the cluster. These batches constitute a linear, per-replica log; thus, each replica maintains its own append-only log while tracking the global log status of all replicas. A batch secures a \textit{Proof of Availability (PoA)}  once it has been persisted and acknowledged by a write-quorum of at least $f+1$ replicas. Then, this PoA is broadcast to the cluster to confirm durability.

    \item \textbf{Epoch-based Global Commit} (\S~\ref{sec:commit}):
    A designated \textit{coordinator} periodically initiates a \textit{new epoch} to commit available batches. This is achieved through a lightweight consensus round (e.g., using Raft~\cite{diegoRaft2014} or Paxos~\cite{lamportPaxosMadeSimple2001}) to establish a "consistent cut" (the vector of log heads) across the per-replica logs. The agreement guarantees that all replicas commit the exact same set of batches in the same order. Crucially, unlike traditional synchronous replication, this mechanism does not block the continuous replication of new batches, effectively decoupling the consensus critical path from the data propagation path.

    \item \textbf{Deterministic Validation and Re-execution} (\S~\ref{sec:valid}): 
    At the epoch boundary, each replica \textit{independently} aggregates the referenced batches and performs validation to detect read-write or write-write conflicts. To minimize aborts, a deterministic algorithm identifies the maximum subset of non-conflicting transactions using a MWIS solver. Transactions excluded from the validated set are then re-executed using deterministic locking~\cite{thomsonCalvinFastDistributed2012}. Once the results of both the validated and re-executed transactions are persisted, the database advances to a new \textit{consistent snapshot} for the epoch, and the clients are acknowledged.
\end{enumerate}

\subsection{Optimistic Local Execution}
\label{sec:local}

\begin{algorithm}[h]
    \fontsize{9pt}{10pt}\selectfont
    \caption{Optimistic Transaction Execution at a Local Replica}\label{code:local}
    \begin{algorithmic}[1]
        
        \State var \textbf{\textit{current\_bid}} : int
        \State var \textbf{\textit{current\_batch}} : Batch
        \State var \textbf{\textit{temp\_state}} : DatabaseState $\gets$ Snapshot of the previous epoch
        \Function{ExecuteTransaction}{$txn$}  
            \State OCC read $temp\_state$
            \If{accessed items are last updated by $txn$s in the same epoch} 
                \State add the $txn.tid$s to $txn.dependencies$ \Comment{Dependencies tracking for conflict detection}
            \ElsIf{accessed items last updated in previous epochs}
                \State add the $eid$s to $txn.read\_keys\_eids$
            \EndIf
            \State OCC validation
            \State OCC write $txn.write\_set$ and update item's $tid$ to $temp\_state$
            \State Store $txn.write\_set$, $txn.read\_set$, $txn.dependencies$, $query$ in $TxnRecord$
            \State $current\_batch.append(TxnRecord)$ 
        \EndFunction

        \Function{BroadcastBatchAsync}{}  \Comment{Runs asynchronously}
        \Loop \ periodically or when $current\_batch$ is full
            \State Broadcast($current\_batch$) 
            \State $current\_bid \gets current\_bid + 1$ \Comment{Increment batch ID}
            \State $current\_batch \gets$ new Batch()
            
        \EndLoop
    \EndFunction

    \end{algorithmic}
\end{algorithm}

As shown in Algorithm~\ref{code:local}, transactions perform optimistic local execution while tracking their dependencies. Upon arrival at a replica (the $ExecuteTransaction()$ function), each transaction is assigned a unique identifier ($tid$) and executed concurrently via a standard OCC protocol as if running on an independent database.

Transactions access a temporary state containing uncommitted local changes applied atop the most recent consistent snapshot during read. This snapshot reflects the database state at the conclusion of the previous epoch. The temporary state is reset to the new snapshot at the start of every new epoch.

Data items buffered at the temporary state are tagged with the $tid$ of the updating transaction. This mechanism is used to track \textit{local dependent transactions}, that is, transactions that read data items updated by other transactions within the same epoch. These dependencies are critical for cascading aborts: if a writer transaction ($T_{write}$) is aborted, any dependent reader ($T_{read}$) must also be aborted, as its read-set has become invalid. Conversely, when a transaction reads an unmodified item directly from the consistent snapshot, it records the item's version, denoted by the epoch ID ($eid$) of its last update. This $eid$ is utilized during the global commit phase (\S\ref{sec:commit}) to detect stale reads (i.e., cases where the local snapshot was superseded by concurrent remote updates).

Following execution, the transaction's input, read-set, and write-set are stored in a $TxnRecord$ and appended to the current batch. Upon reaching capacity or after a timeout, the batch is asynchronously broadcast to all replicas ($BroadcastBatchAsync()$). Simultaneously, these batches are persisted into a linear, per-replica \textit{local log}, ensuring availability for the subsequent commit protocol.

\begin{figure}[h]
    \centerline{\includegraphics[scale=0.70] {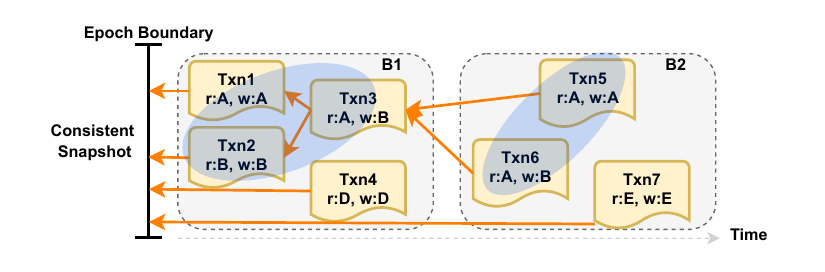}}
    \caption{Transactions and Batches}
    \label{fig:txs}

\end{figure}

Figure~\ref{fig:txs} illustrates the organization of locally executed transactions within the log of batches, where arrows denote data dependencies. In this example, $7$ transactions are distributed across two batches. Transactions $tx1$, $tx2$, $tx4$, and $tx7$ are independent; they operate on disjoint keys and read directly from the consistent snapshot. In contrast, $tx3$, $tx5$, and $tx6$ are dependent transactions, reading values written by uncommitted local operations from the temporary state. Specifically, $tx3$ depends on $tx1$ and $tx2$ (reading data item $A$ and writing to $B$), which implies an execution order where $tx3$ succeeds its predecessors. Similarly, $tx5$ depends on earlier writes. Transaction $tx6$, reading both $A$ and $B$, depends on both $tx5$ and $tx3$. The blue circles encircle these dependency groups, which we designate as \textit{Transaction Chains (TxnC)}. These chains serve as the atomic units for conflict detection during the global validation phase, as discussed in detail in Section~\ref{sec:valid}.

\subsection{Asynchronous Replication}
\label{sec:replication}

\begin{algorithm}[h]
    \fontsize{9pt}{10pt}\selectfont
    \caption{Batch Dissemination}\label{code:broadcast}
    \begin{algorithmic}[1]
        \State var \textbf{\textit{logs}} : \{$rid$ : int -> list of Batches\} \Comment{a map of replica IDs to the list of batches that have been received at the local replica}
        \State var \textbf{\textit{PoA\_heads}} : \{$rid$ : int -> int\}

        \Function{ReceivedBatch}{$batch$ : list, $rid$}  
            \State $logs[rid][batch.bid] \gets batch$
            \State send $ack$ to $rid$
        \EndFunction

        \Function{ReceivedAcknowledgement}{$ack$}
            \State var \textbf{\textit{local\_log}} : list of Batches $\gets logs[self.rid]$
            \State $bid\_acks[ack.bid] \gets bid\_acks[ack.bid] + 1$ \Comment{Assume ack's sender has been checked for duplicates}

            \If{$bid\_acks[ack.bid] \geq f+1$}
                \State $local\_log[ack.bid].available \gets true$ \label{line: available}
                
                \For{$i \gets PoA\_heads[self.rid] + 1$ to $ack.bid$} \Comment{If a batch and its predecessors are confirmed available, broadcast the PoA until the latest available one.} \label{line: poa}
                    \If{$local\_log[i].available$ \&\& $local\_log[i - 1].sent\_poa$} 
                        \State BroadcastPoA($i$) \
                        \State $local\_log[i].sent\_poa \gets true$ 
                    \Else
                        \State break \Comment{No need to check further}
                    \EndIf 
                \EndFor 
            \EndIf
        \EndFunction
        
        \Function{ReceivedPoA}{$PoA$}
            \If{logs[PoA.rid][PoA.bid] != null}
                \State $logs[PoA.rid][PoA.bid].PoA\_sent \gets true$ \Comment{The batch gets PoA}
                \State $PoA\_heads[PoA.rid] \gets \max(PoA\_heads[PoA.rid], PoA.bid)$
            \Else
                \State BroadcastRequestBatch($PoA.rid$, $PoA.bid$)
            \EndIf
        \EndFunction
  
    \end{algorithmic}
\end{algorithm}

Algorithm~\ref{code:broadcast} shows the replication process. Each replica maintains its own log of locally generated batches and also tracks a global view ($logs$) of the logs of batches. The indices of these logs correspond to the batch IDs ($bid$) from their source replicas (the replica that created the batch). Each batch progresses through four states: \textit{broadcast}, \textit{available}, \textit{PoA\_sent}, and \textit{committed}. Figure~\ref{fig:log} illustrates this multi-log view from the perspective of Replica A.

\begin{figure}[h]
    \centerline{\includegraphics[scale=0.8] {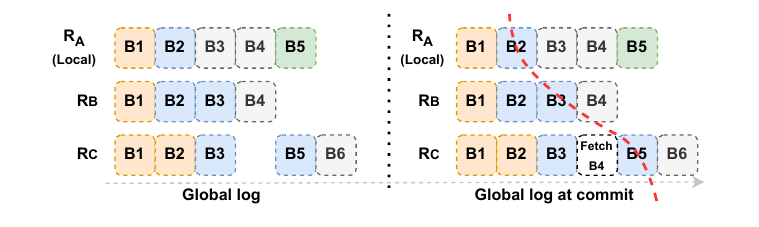}}
    \caption{Replica A's view of the logs: Brown means committed, blue means available and has PoA, green means broadcast but without PoA, and white means received but not available. The red dotted line is the coordinator's consistent cut.}
    \label{fig:log}
\end{figure}

Upon receiving a batch from a peer, a replica appends the entry to the corresponding source replica's log, marks it as \textit{broadcast}, and sends an acknowledgment ($ReceivedBatch()$). The source replica aggregates these acknowledgments in $ReceivedAcknowledgement()$. After securing a quorum of $f+1$ acknowledgments, the batch transitions to the \textit{available} state (Line~\ref{line: available}). At this point, the source replica broadcasts a \textit{Proof of Availability} (PoA), but only when all preceding batches in the sequence have already reached the \textit{available} state and issued their own PoAs. This constraint ensures that availability is certified only for contiguous log prefixes. The PoA serves as a guarantee that the batch is persisted on at least one correct replica, satisfying the $f+1$ availability requirement.

The PoA generation process (Line~\ref{line: poa}) iteratively scans the log from the current $PoA\_head$ (the index of the last batch with a transmitted PoA) up to the newly available batch. Again, to ensure log contiguity, a PoA is broadcast for a batch only if all its predecessors are also marked \textit{available} and the $PoA\_head$ is advanced. Upon receiving a PoA ($ReceivedPoA()$), peer replicas update the corresponding batch's state to $PoA\_sent$. In addition, if a replica receives a PoA for a batch that is absent from its local log, it broadcasts a fetch request to retrieve the missing batch.

In the example shown in Figure~\ref{fig:log} (left), batch $B5$ on Replica A is marked $available$ after securing $f+1$ acknowledgments. However, its predecessors, $B3$ and $B4$ remain pending. Thus, the contiguity constraint prevents Replica A from sending a PoA for $B5$ until the gaps are filled. Note that the $available$ state is internal to the source replica (which tracks quorums); peer replicas only observe remote batches in the $broadcast$, $PoA\_sent$, or $committed$ states.

\subsection{Global Epoch Commit}
\label{sec:commit}

\begin{algorithm}[h]
    \fontsize{9pt}{10pt}\selectfont
    \caption{Consistent Cut}\label{code:consensus}
    \begin{algorithmic}[1]
        \State var \textbf{\textit{last\_committed}} : \{$rid$ : int -> int\}

        \Function{StartConsensus}{}  \Comment{At the coordinator}
            \Loop \ periodically
                \State ProposeValue($PoA\_heads$)
            \EndLoop
        \EndFunction

        \Function{AgreementReached}{$cut\_bids$} \Comment{At all replicas}  
            \State var \textbf{\textit{batch\_to\_commit}} : list of batches $\gets []$

            \ForAll{$rid$}
                \For{$i \gets last\_committed[rid]$ to $cut\_bids[rid]$} \label{line:getbatches}

                    \If{$logs[rid][i] == null$}
                        \State $logs[rid][i] \gets $ BroadcastRequestBatch($rid$, $i$)
                    \EndIf

                    \State $batch\_to\_commit.append(logs[rid][i])$
                \EndFor
            \EndFor

            \State Validation($batch\_to\_commit$)
        \EndFunction

    \end{algorithmic}
\end{algorithm}

Periodically (every \textit{epoch}), a designated \textit{coordinator} replica (typically the Raft leader) initiates a consensus round triggered by a local timer ($StartConsensus()$). The coordinator proposes a "consistent cut" of the global log, consisting of the vector of batch indices currently marked as \textit{PoA\_sent} in its local view ($PoA\_heads$). Standard consensus protocols, such as Raft~\cite{diegoRaft2014}, are used for this proposal. Unlike systems that depend on synchronized physical clocks or centralized sequencers for ordering~\cite{thomsonCalvinFastDistributed2012,fanOceanVistaGossipbased2019, luEpochbasedCommitReplication2021, zhougeogauss2023}, \scmacro{} reaches agreement directly on the set of available batches. This approach establishes a consistent global checkpoint without assuming tight clock synchronization or imposing hard real-time constraints.

Given the monotonic nature of $bid$, this consensus step yields a uniform agreement on a specific vector of indices, effectively making a "cut" across the global logs. The consistent cut deterministically identifies the commit scope for the current epoch: specifically, all batches positioned between the previously established commit frontier and the new cut (Line~\ref{line:getbatches}). Because the inclusion of a $bid$ in the cut implies the existence of a PoA, and by extension, the availability of all preceding batches, the integrity of the log sequence is ensured. Due to the $f+1$ availability enforced by the PoA, any gaps in the log remain retrievable from at least one correct replica.

Figure~\ref{fig:log} (right) illustrates a consistent cut from the perspective of Replica A (also the coordinator). It identifies the most recent batches with PoAs: batch $B2$, batch $B3$ from Replica B, and batch $B5$ from Replica C. Based on this frontier, it proposes the vector $\{2, 3, 5\}$ as the global commit boundary for the current epoch. Notably, although Replica A has not yet received batch $B4$ from Replica C, the existence of a PoA for $B5$ transitively guarantees that all preceding batches, including $B4$, are safely persisted on a quorum. Replica A will therefore attempt to retrieve the missing batch.

The requirement of PoA's predecessor availability (\S\ref{sec:replication}) is a deliberate design choice for two reasons. First, it minimizes consensus overhead: by guaranteeing log contiguity, replicas only need to agree on a compact vector of log heads (the "cut"), rather than arbitrary lists of batch IDs. Second, and more importantly, this constraint enforces a commit order that mirrors the original transaction order. Preserving this sequence maintains the integrity of the local dependency chains established during execution, which significantly simplifies the logic required to enforce serializability.

\subsection{Validation and Deterministic Re-execution}
\label{sec:valid}

\begin{algorithm}[h]
    \fontsize{9pt}{10pt}\selectfont
    \caption{Validation and Commit}\label{code:commit}
    \begin{algorithmic}[1]
        \State var \textbf{\textit{current\_eid}} : int \Comment{Index of the current epoch}
        
        \Function{Validation}{$batches$}
            \State var (\textbf{\textit{invalid\_txns}}, \textbf{\textit{valid\_txns}})  = \\FilterConflicts($batches$) \Comment{See \S~\ref{sec:filter}}
            \State $durable\_storage.update(valid\_txns.write\_sets)$
            \State Sort $invalid\_txns$ based on the deterministic rules
            \State Execute $invalid\_txns$ with a deterministic lock
            \State Update $eid$ for all written items with $current\_eid$
            \State $current\_eid \gets current\_eid + 1$
        \EndFunction

    \end{algorithmic}
\end{algorithm}

After collecting the $batches\_to\_commit$, replicas independently complete the validation and re-execution phase without further coordination.

We define two types of invalid transactions. The first type is \textit{stale transactions}: those that read a data version from a consistent snapshot that has been superseded by a committed update prior to the current epoch. We use the epoch ID ($eid$) as the write version to detect stale reads (Line~\ref{code:stale} in Algorithm~\ref{code:filter}). In addition, we enforce cascading invalidation: any transaction that depends on a previously invalidated transaction, even across epochs, is recursively marked as stale. The second type is conflicting transactions, which include read-write or write-write conflicts in the current epoch. 

Conflict detection operates at the granularity of \textit{Transaction Chains} (TxnC), as defined in Section~\ref{sec:local}. Under this model, chains are treated as atomic units: a conflict involving any single transaction invalidates the entire chain. We adopt this coarse-grained approach because fine-grained resolution, such as dynamically identifying and pruning only specific conflicting dependencies, incurs computational complexity that often outweighs the performance benefits of salvaging a small subset of transactions, particularly when the cascade of dependent aborts is unpredictable~\cite{dongFineGrainedReExecutionEfficient2023}. An important property to note is that transaction chains contain only transactions with the same source replica, and conflicts can only arise between chains from different replicas, because OCC already ensures local serializability within batches from the same replica. In order to reduce the number of conflicting transactions that need to be re-executed, we employ a graph-based conflict detection algorithm as described in Section~\ref{sec:filter}.

Following conflict resolution, the write-sets of all valid transactions are committed to durable storage. A key aspect of our algorithm is to \textit{deterministically re-execute} the conflicting transactions instead of aborting them. To ensure serializability, stale and conflicting transactions are first arranged into a total order determined by their $tid$ and the predefined priority of their source replica. This queue is then executed using a deterministic locking manager, derived from Calvin~\cite{thomsonCalvinFastDistributed2012}, to enforce serializable execution according to the sorted order while maintaining concurrency. Deterministic locking requires prior knowledge of the read and write sets, which are available in our design because they are tracked during the initial OCC phase. Once the results of the re-executed transactions are persisted, the system responds to clients with the results.

\subsection{Detecting Conflicting Transactions}
\label{sec:filter}
Given a batch of transactions, we must efficiently identify write-write and read-write conflicts among them to ensure serializability~\cite{bernsteinconcurrency1981,ceriOCC1982}. The conflict detection process must be \textit{deterministic}, ensuring that every replica independently reaches the exact same commit decision; \textit{optimal}, minimizing the number of re-executions; and \textit{computationally efficient}, ensuring that the validation logic does not become a bottleneck in the commit critical path.

We formulate this optimization as finding the maximum weight independent set~\cite{liangmwis1991} on a \textit{conflict graph}~\cite{bernsteinconcurrency1981}. In this model, each vertex represents a transaction chain, assigned a weight corresponding to the number of transactions within. An undirected edge links any two vertices if a transaction in one chain exhibits a data conflict with a transaction in the other. The objective is to identify the independent set (a subset of non-adjacent vertices) with the highest weight. This set corresponds to the maximum number of transactions that can be safely committed in parallel. By maximizing the weight of the committed set, the system effectively minimizes the number of transactions to be re-executed.

\begin{algorithm}[h]
    \fontsize{9pt}{10pt}\selectfont
    \caption{Conflict Graph Construction and Conflict Detection}\label{code:filter}
    \begin{algorithmic}[1]
        \Function{CreateConflictGraph}{$batch$}
            \State var \textbf{\textit{CG}} : graph
            \State var \textbf{\textit{write\_tracker}} : {data item : unique set of transactions}
            \State var \textbf{\textit{read\_tracker}} : {data item : unique set of transactions}
            \State var \textbf{\textit{stale\_txn}} : []

            \State Group dependent $txn$s from the same batches into \text{transaction chains} $txnc$.

            \ForAll{$txnc$ in $batch$} 
                \If{$\exists item \in txnc.read\_set$ s.t $item.eid < current.eid$ or $\exists txn \in txnc$ s.t $txn.prev$ is aborted} \Comment{Check for staleness or cascading aborts} \label{code:stale}
                    \State $stale\_txn.append(txnc)$
                    \State continue
                \EndIf

                \ForAll{$item$ in $txnc.write\_set$}
                    \State $write\_tracker[item].add(txnc)$
                \EndFor

                \ForAll{$item$ in $txnc.read\_set$}
                    \State $read\_tracker[item].add(txnc)$ 
                \EndFor
                \State $CG.add\_vertex(txnc)$ 
            \EndFor

            \ForAll{$key, value$ in $write\_tracker$} 
                \If{$value.size > 1$} \Comment{Write-write conflicts}
                    \State add edges between all transactions in $value$ to $CG$
                \EndIf
            \EndFor

            \ForAll{$key, value$ in $read\_tracker$} \Comment{Read-write conflicts}
                \If{$write\_tracker[key] \neq null$} 
                    \State  add edges between all transactions in $value$ and write\_tracker[key] to $CG$
                \EndIf
            \EndFor

            \State return $CG$, $stale\_txn$ 
        \EndFunction

        \Function{FilterConflicts}{$batch$}
            \State \textbf{var} $CG, stale\_txns \gets$ \Call{CreateConflictGraph}{$batch$}
            \State \textbf{var} $valid\_set \gets$ \Call{SolveMWIS}{$CG$} 
            \State \textbf{var} $conflicting\_set \gets CG.vertices \setminus valid\_set$
            
            \State \textbf{return} $(conflicting\_set \cup stale\_txns, \ valid\_set)$
        \EndFunction
    \end{algorithmic}
\end{algorithm}

\begin{figure}[h]
    \centerline{\includegraphics[scale=0.8] {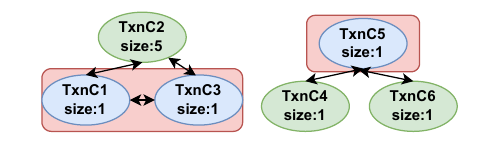}}
    \caption{Conflict graph of 6 transaction chains, arrows mean read-write or write-write conflicts, and the red box shows the conflicting transaction chains ($TxnC$) that are invalidated and to be re-executed.}
    \label{fig:cg}

\end{figure}

Algorithm~\ref{code:filter} shows the approach. We first construct the conflict graph ($CreateConflictGraph()$) through an iterative scan of all write-sets, linking transaction chains that have read-write or write-write contention on shared data items. Subsequently, the resulting graph is processed by the MWIS solver to identify the optimal subset of non-conflicting transactions. The chains that are excluded from the independent set are designated as conflicting. In cases where multiple optimal solutions exist, the solver may employ a deterministic tie-breaking rule, such as prioritizing the transaction chain with the lowest starting $tid$.

An example is shown in Figure~\ref{fig:cg}, where $txnc2$ represents a heavy chain containing 5 transactions, while all other chains ($txnc1$, $txnc3$--$txnc6$) each contain a single transaction. $txnc1$, $txnc2$, and $txnc3$ are in conflict because they may all have accessed the same data item. However, the MWIS solver should remove both $txnc1$ and $txnc3$ because $txnc2$ has the highest weight. For $txnc4$, $txnc5$, and $txnc6$, a possible scenario is that there are two data items, $a$ and $b$: $txnc4$ writes to $a$, $txnc5$ writes to both $a$ and $b$, and $txnc6$ reads from $b$. In this case, removing $txnc5$ is optimal.

Another issue to consider is cross-epoch and cross-batch dependent transactions. For cross-epoch dependencies, if a transaction depends on other transactions but they are not committed in the same epoch, then the transaction in the second batch must be derived from the previous consistent snapshot, so it is treated the same as a stale transaction. For dependencies across batches but still within the same epoch, the transactions are placed in the same global batch, and the links between the transaction chains across batches can be used to merge them into a single chain when constructing the conflict graph.

\subsection{Correctness}
We show that \scmacro{} guarantees One-Copy Serializability (1SR)~\cite{bernsteinconcurrency1981} under crash failures. That is, all replicas observe the same serializable execution order of transactions. A transaction is defined as a series of read, write, or delete operations, where each operation accesses only a single data item. We do not assume any range queries or transactions with unbounded read/write sets, which can be supported with additional mechanisms such as index or partition versioning, or client-side locks~\cite{levandoskiMultiVersionRangeConcurrency2015}.

\begin{lemma}~\label{l:poa}
If a Proof of Availability (PoA) is received by a correct replica, then at least one correct replica has received and stored the corresponding batch.
\end{lemma}

\begin{proof}
    A PoA is only generated when a batch has received $f+1$ acknowledgements (Line~\ref{line: available} of Algorithm~\ref{code:broadcast}). Since the system can tolerate up to $f$ crash failures, at least one replica that has persisted the batch must be correct.
\end{proof}

\begin{lemma}~\label{l:agree}
At each epoch, all correct replicas observe the same set of batches in the same order.
\end{lemma}

\begin{proof}
The coordinator only proposes the consistent cut with indices of batches that have received a PoA. Furthermore, a replica only broadcasts a PoA for a batch at index $i$ if all predecessor batches (indices $< i$) are already marked as available and have sent their PoAs (Line~\ref{line: poa} of Algorithm~\ref{code:broadcast}). Therefore, by induction, if the coordinator observes a PoA for index $j$ (the cut index), then it is guaranteed that PoAs (and thus availability) exist for all indices $0, 1, \dots, j-1$.

The consensus protocol (e.g., Raft) guarantees that the consistent cut indices are agreed upon by all correct replicas in the same order. Since any batches with indices less than the cut index are guaranteed to be available by Lemma~\ref{l:poa}, all these batches are guaranteed to be available. Because the batches are uniquely identified by monotonically increasing indices, their order will be identical across all correct replicas.
\end{proof}

\begin{lemma}~\label{l:ex}
At each epoch, the same set of transactions is ordered to be executed in a serializable manner among all correct replicas.
\end{lemma}

\begin{proof}
By Lemma~\ref{l:agree}, all correct replicas observe the same set of batches to commit, and thus the same set of all transactions $S_{all}$. The $CreateConflictGraph()$ function in Algorithm~\ref{code:filter} is deterministic, as it detects all read-write and write-write conflicts among transactions (chains). The MWIS solver is also deterministic and returns the same optimal set of transactions $S_{opt}$ without read-write or write-write conflicts at each replica. Thus, there is no dependency among these transactions, and they can be considered serializable~\cite{thomasianDistributedOptimisticConcurrency1998}.

The remaining transactions $S_{re}$ are re-executed using a deterministic locking mechanism. By the definition of deterministic locking~\cite{thomsonCalvinFastDistributed2012}, $S_{re}$ is also deterministically serialized among all replicas.

Since $S_{re}$ is always executed after $S_{opt}$, and $S_{re} \cup S_{opt} = S_{all}$, the execution order of $S_{all}$ is serializable.
\end{proof}

\begin{definition}
Define a single-copy serial history $S_{ser}$ consisting exactly of all transactions committed from epoch $1$ to $e$, such that executing $S_{ser}$ on an initial database state of any replica produces the same final database state.
\end{definition}

\begin{theorem}~\label{l:serial}
\scmacro{} guarantees One-Copy Serializability (1SR).
\end{theorem}

\begin{proof} 
\textit{Base Case:} Let $E_0$ be the initial epoch; it has no transactions and is trivially serializable.

\textit{Induction Hypothesis:} Assume that up to epoch $E_{k - 1}$, all transactions are executed in a serializable manner.

\textit{Induction Step:} By the induction hypothesis, at the beginning of $E_k$, all replicas must start in the same state (the consistent snapshot). Thus, all transactions in epoch $E_k$ are validated and committed against the same consistent snapshot.

By Lemma~\ref{l:ex}, the commit procedure produces an execution order $S_k$ of all transactions in $E_k$ equivalent to a serial execution at all replicas, such that the state at the end of the epoch is $E_{k} = E_{k-1} + S_k$.

Since the state at the end of epoch $E_{k-1}$ is consistent among all replicas, and $S_{re}$ is serializable, the state at the end of epoch $E_k$ is also consistent among all replicas. Thus, $\sum_{i=1}^{k} S_i$ is $S_{ser}$.
\end{proof}

\section{Implementation}
\label{sec:impl}
We implement \scmacro{}~\cite{minerva} as a distributed transactional key-value store. The system is developed in C\# using .NET 9.0. For the network layer, we employ raw TCP sockets to handle all inter-replica and client interactions. We utilize MemoryPack~\cite{memorypack}, a high-performance binary serializer for the communication protocol.

\paragraph{OCC and Deterministic Locking}
Replicas accept client transaction requests via raw TCP sockets. Upon arrival, requests are dispatched to worker threads managed by the .NET Task Parallel Library for optimistic execution, as detailed in \S\ref{sec:local}. Once local results are batched, the execution task suspends and awaits the global commit signal. We leverage .NET's asynchronous programming model (async/await) to ensure these waits are non-blocking, thereby minimizing context-switching overhead and preventing thread pool exhaustion during high-latency commit windows.

Our deterministic re-execution engine implements a deterministic locking scheduler derived from Calvin~\cite{thomsonCalvinFastDistributed2012}. Transactions are first assigned a global order based on a tuple of the source replica's static priority and the transaction's original ID ($tid$). A centralized lock manager enforces this schedule using per-key wait queues. Since the read and write sets are already known from the initial OCC attempt, lock acquisition is deterministic and deadlock-free. The lock manager maintains a dependency counter for each pending transaction; as active transactions complete and release their locks, the manager processes the relevant queues. When a transaction's dependency counter reaches zero, it is scheduled for execution on a dedicated worker thread.

\paragraph{Transaction Chains Building and MWIS Solving}
As established in \S\ref{sec:valid}, we first need to identify all transaction chains within batches from the same replica. Using the local dependencies that are tracked during the OCC phase, the task is to find all \textit{connected components} within the local dependency graph (distinct from the global conflict graph). This graph treats each transaction as a vertex, with edges linking transactions that share a local dependency. Then, a depth-first traversal partitions the graph into disjoint sets, producing the transaction chains.

Next, we use an open-source C++ graph library~\cite{Solver} that provides both exact and approximate solvers for the MWIS problem. For the exact solution, an integer linear programming (ILP) optimization method~\cite{wolseyIntegerProgramming2020} is used. Although this approach incurs higher computational latency, it guarantees the identification of the optimal independent set and minimizes the volume of transactions to re-execute. The approximate algorithm is based on a popular greedy heuristic~\cite{sakaigreedymwis2003}, which prioritizes decision speed over optimality. The performance trade-offs between these two approaches are analyzed in Section~\ref{sec:skew}.

\paragraph{Fault Tolerance}
\scmacro{} eliminates the need for the blocking 2PC process to ensure atomicity within the replication cluster because of the replication protocol and deterministic validation. Replicas can thus independently converge on the same commit or abort decisions for every transaction without additional coordination. 

The Raft cluster is also used as the configuration manager. The Raft leader assumes the role of the coordinator replica, which is responsible for initiating the periodic global commit process. This also effectively handles coordinator failures, as Raft guarantees that a new leader will be elected to resume the commit process. We use an open-source implementation from the .Next library~\cite{DotNext}.

Finally, we propose the following potential recovery mechanisms to handle replica failures. While failed replicas can recover by fetching the missing batches from other replicas based on the global log, the global log should be persisted on the disk with a write-ahead log, allowing the recovered replica to only gather the missing batches without the need to re-execute all transactions within the epoch. Furthermore, storage nodes can employ a localized primary-backup mechanism with asynchronous backup. Since transactions are always written to the database snapshot first, we do not need to consider the consistency problem with the backups by using epoch as versioning.

\paragraph{High-contention Mode}
We implement an adaptive \textit{high-contention mode} to reduce performance impact under heavily contended workloads. This mode is triggered when the percentage of re-executed transactions exceeds a configurable threshold (default $50\%$) for a sustained duration (default $200$ ms). In this mode, replicas suspend the optimistic execution phase entirely. Instead, incoming transactions are statically probed to extract their read and write sets and are immediately forwarded to the current batch for replication. During the global commit phase, these transactions bypass the conflict detection phase and are automatically routed to the deterministic execution queue.

High-contention mode essentially transforms the system into a pure deterministic scheduler (similar to Calvin) during conflict spikes. Since a pre-computed write set is likely to be discarded, doing so can reduce wasted computational cycles and significantly reduce network overhead, as the propagation of large write-sets is avoided. However, this optimization incurs a latency penalty, as transaction execution must wait until the global commit point.

\paragraph{Garbage Collection}
To prevent unbounded storage growth, obsolete log entries are garbage collected when they are no longer needed, that is after they have been fully processed by the entire cluster. To track this, replicas periodically broadcast their progress (their latest committed epoch). Each node maintains a local view of the cluster's "low watermark," which is the highest epoch index that has been universally committed. Once this watermark advances beyond epoch $i$, replicas safely truncate the log up to that point. In addition, the local $temp\_state$ is cleared at each epoch boundary, effectively garbage collecting the invalid transactions.

\section{Evaluation}
\label{sec:eval}
\subsection{Experimental Setup}
We deploy the \scmacro{} cluster on the Compute Canada Cloud. Each replica operates as a monolithic process, colocating storage and execution components within a single Virtual Machine (VM). Each node is provisioned with 4 vCPUs (Intel Xeon Skylake) and 48 GB of RAM. The operating system is Ubuntu 24.04, except for the GeoGauss baseline, which requires CentOS 7.6. All nodes interact via a high-speed LAN with 3 Gbps bandwidth. To prevent client-side bottlenecks, client threads run on a separate 16-core VM.

Regarding \scmacro{}'s configuration, the epoch interval is set to 15 ms. Batch transmission is triggered by either reaching a size threshold of 4 MB or a timeout of 5 ms, whichever occurs first. Unless otherwise specified, the adaptive high-contention mode is enabled with default settings, and the exact MWIS solver is used.

We compare \scmacro{} against state-of-the-art systems that share key features: (1) Geo-replicated databases, including GeoGauss~\cite{zhougeogauss2023}, Ocean Vista~\cite{fanOceanVistaGossipbased2019}, and CockroachDB~\cite{taftCockroachDBResilientGeoDistributed2020a}; (2) epoch-based systems, including COCO~\cite{luEpochbasedCommitReplication2021}; and (3) deterministic databases, including CalvinDB~\cite{thomsonCalvinFastDistributed2012, CalvinDB}.

For GeoGauss, COCO, and Calvin, we configure all to use the default $10ms$ epoch time as suggested in their papers. We also use COCO's \textit{logical clock} mode, which offers better performance. For a fair comparison, we set all systems for full database replication by setting the number of replicas equal to the number of VM nodes, regardless of their individual sharding capabilities. Finally, all systems are set to in-memory mode, with disk persistency disabled when possible.

The source code for Ocean Vista is not publicly available, so we reimplement its protocol using \scmacro{}'s network and storage layer. While the original design relies on an epidemic gossip protocol for metadata dissemination, we simulate this behavior using a simplified periodic batch broadcast mechanism. Additionally, we replace its reliance on a synchronized clock using Network Time Protocol (NTP) with a logical deterministic ordering scheme derived from local clocks and replica IDs.

Finally, we note that the open-source implementation of GeoGauss supports only Snapshot Isolation (SI). Although this suffices for TPC-C correctness, it represents a weaker guarantee compared to the serializability enforced by \scmacro{} and the other systems.

\subsection{Benchmarks and Metrics}
We evaluate the systems with two widely adopted database benchmarks: YCSB~\cite{cooper2010YCSB} and TPC-C~\cite{TPCC}. YCSB is a simple read-write benchmark for key-value stores. In our experiments, the database consists of 2 million records, with 10-byte keys and 1 KB values. The workload is based on YCSB-A (50\% reads, 50\% writes) where each transaction contains 10 operations (reads or writes) on random keys. Because COCO and Calvin do not natively support blind writes, we adapt the workload for these baselines to perform a read-before-write for every update to accommodate their specific transaction models. Unless otherwise noted, we set the Zipfian skew factor to 0 to simulate a uniform access pattern. This workload stresses the network and data access since each transaction on average contains $5 KB$ of write payloads.

TPC-C is a sophisticated OLTP benchmark with more complex transactions that simulates the operation of a warehouse-based ordering system. We configure a 100-warehouse setup with $50\%$ of \textit{Payment} transactions and $50\%$ of \textit{New-Order} transactions. These are the two primary transaction types in the standard TPC-C benchmark, and this workload is used for evaluating all baselines except for CalvinDB. It is more computationally heavy and exhibits more contention than YCSB, but has smaller payload sizes.

The main performance metrics are \textit{peak throughput (txns/s)} and \textit{end-to-end latency (ms)}. Both metrics reflect only the \textit{successfully committed} transactions on the \textit{client} side. To determine peak throughput, we incrementally scale the client load until system saturation is observed. We then compute the average of the highest sustainable throughput across five 30-second experiment runs. The specific client concurrency is tuned for each configuration to identify the maximum saturation point, ranging from 2000 to 8000 total client threads.

\subsection{Effects of Cross-replica Network Latency}
\begin{figure}[h]
    \centerline{\includegraphics[scale=0.5] {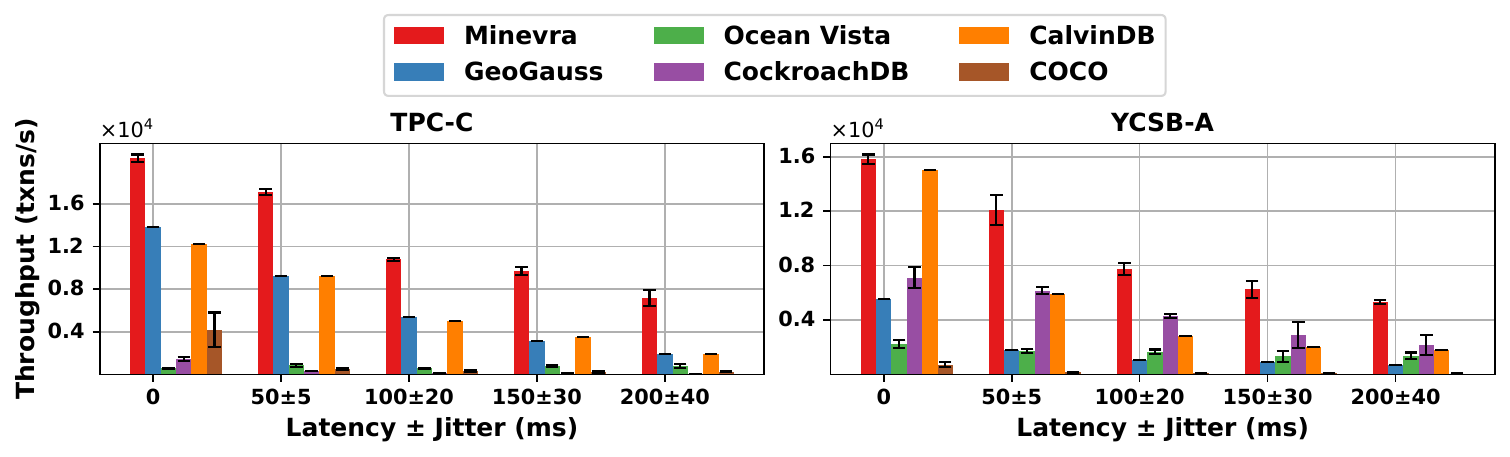}}
    \caption{Throughput vs. Cross-replica Network Latency}
    \label{fig:lt}
\end{figure}

We first evaluate the performance of the systems in geo-distributed scenarios by introducing synthetic network delay among 3 replicas using Linux's \textit{tc} and \textit{netem} tools. We vary the delay from $50ms$ to $200ms$, adding $5ms$ to $40ms$ of jitter to simulate various cross-region deployment conditions.

\paragraph{Throughput} Figure~\ref{fig:lt} illustrates the throughput results. In a low-delay LAN environment, \scmacro{}, CalvinDB, and GeoGauss achieve the highest performance levels, each exceeding $10,000$ txns/s; notably, \scmacro{} surpasses the latter two by approximately $30\%$. In contrast, COCO, CockroachDB, and Ocean Vista exhibit significantly lower throughput, below $5,000$ txns/s, indicating architectural limitations under full replication. For instance, COCO is optimized for sharded execution; this configuration effectively treats every transaction as a cross-partition transaction, requiring 2PC among all replicas and limiting its performance optimizations.

Introducing $50$ ms of cross-replica latency degrades performance globally. At this point, GeoGauss loses $33\%$ of its throughput and CalvinDB loses $40\%$, whereas \scmacro{}'s throughput decreases by only about $15\%$, demonstrating the efficiency of its asynchronous replication and decoupled commit protocol. At $200ms$ delay, \scmacro{} outperforms CalvinDB and GeoGauss by $3.8\times$. 

With YCSB-A, \scmacro{} only has a slight advantage over CalvinDB without network delay. However, as the delay increases, \scmacro{} again shows superior performance, achieving $2.5\times$ the throughput of CockroachDB, the second-best baseline at $200ms$ latency.

We observe a notable divergence in workload sensitivity: \scmacro{}, GeoGauss, and COCO perform worse on YCSB relative to TPC-C, whereas Ocean Vista, CockroachDB, and CalvinDB show improved throughput. This indicates that the former group is highly sensitive to payload composition; specifically, \scmacro{} incurs overhead from the inefficient serialization of YCSB's string payloads. Regarding Ocean Vista, while its gossip protocol makes it less affected by network delay across both benchmarks, its absolute throughput remains low due to high contention, especially in complex transaction scenarios like TPC-C.

\begin{figure}[h]
    \centerline{\includegraphics[scale=0.5] {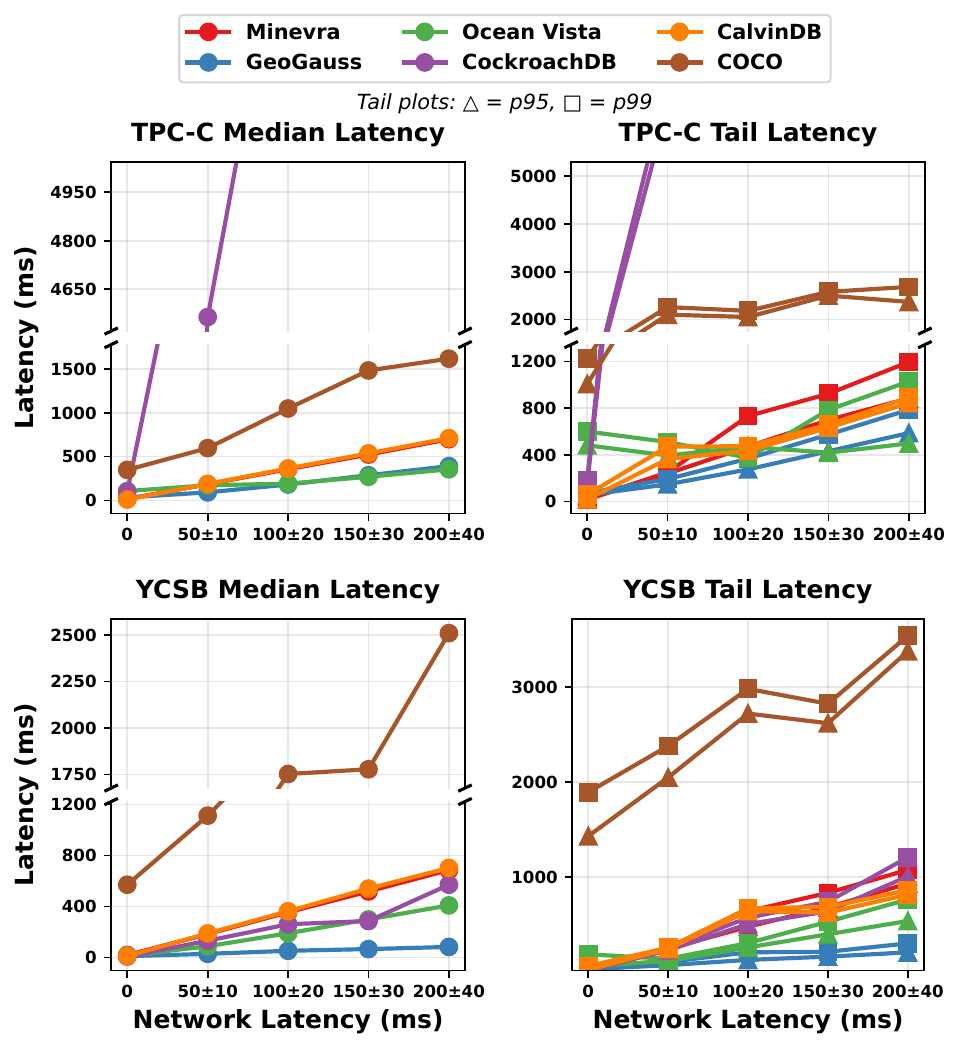}}
    \caption{Transaction Commit Latency  vs. Cross-replica Network Latency}
    \label{fig:ltvnet}
\end{figure}

\paragraph{Latency} Figure~\ref{fig:ltvnet} depicts the end-to-end transaction latency for TPC-C and YCSB-A under low system load, showing the execution latency of individual transactions.

In the TPC-C workload, \scmacro{}, CalvinDB, and GeoGauss exhibit comparable median latencies. For \scmacro{}, this occurs because the consensus delay eventually exceeds the configured epoch time, becoming the new bottleneck, as each commit must wait for the consensus. This extended commit duration also increases the probability of stale transactions, leading to a high re-execution rate, which makes \scmacro{} behave similarly to a deterministic database. For CalvinDB, the increased latency directly impacts its Paxos-based replication time, which is on the critical path. As a result, both systems scale similarly under high network delay.

GeoGauss and Ocean Vista exhibit lower median latencies than \scmacro{} and CalvinDB, effectively achieving single-round-trip commits whenever a replica is able to promptly receive commit information from all other replicas. Transaction commit latency grows linearly with increasing cross-replica latency for all systems except CockroachDB, which is severely affected by network latency among replicas. 

In YCSB, all systems exhibit similar trends, and CockroachDB is comparable to other systems again. Notably, Ocean Vista exhibits elevated tail latencies (p95 and p99); this instability persists due to the protocol's inherent processing overhead. Ultimately, these results showcase the trade-off within \scmacro{}: while the batching mechanism incurs a latency penalty for individual transactions, it preserves robust aggregate throughput, particularly under high network delay.

\subsection{Scalability}

\begin{figure}[h]
    \centerline{\includegraphics[scale=0.5] {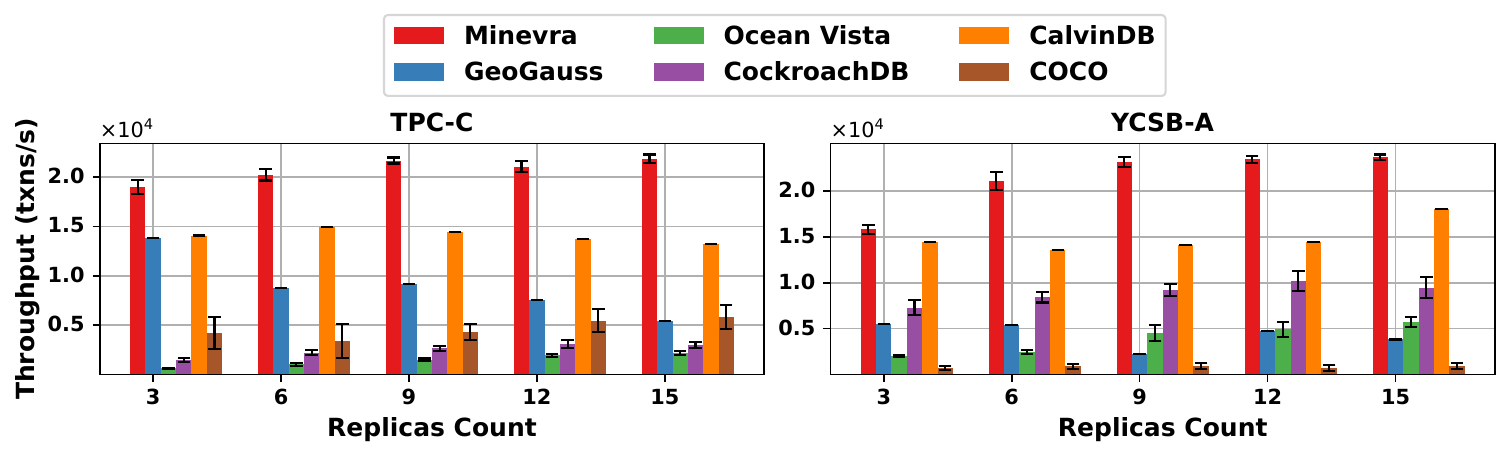}}
    \caption{Scalability in Throughput}
    \label{fig:scale}
\end{figure}

We vary the number of replicas from 3 to 15 and measure the resulting throughput and latency to evaluate scalability.

\paragraph{Throughput} We observe that \scmacro{} maintains a superior scalability profile compared to the competing systems in Figure~\ref{fig:scale}. Across both YCSB-A and TPC-C benchmarks, \scmacro{} outperforms all baselines, observing steady throughput growth up to 12 replicas. At 15 replicas under TPC-C, \scmacro{} achieves approximately $4\times$ the throughput of GeoGauss, COCO, and CockroachDB, and $1.5\times$ that of CalvinDB. Similarly, in YCSB-A, \scmacro{} outperforms CalvinDB by $1.8\times$ and CockroachDB by $2.1\times$. These metrics validate the efficiency of the underlying replication and commit protocols in large-cluster environments.

Notably, GeoGauss is the only system to experience throughput degradation when scaling from 3 to 15 replicas, whereas all others maintain or improve performance. This is likely due to its all-to-all communication in the synchronization phase, which introduces overhead as the cluster expands.

\begin{figure}[h]
    \centerline{\includegraphics[scale=0.5] {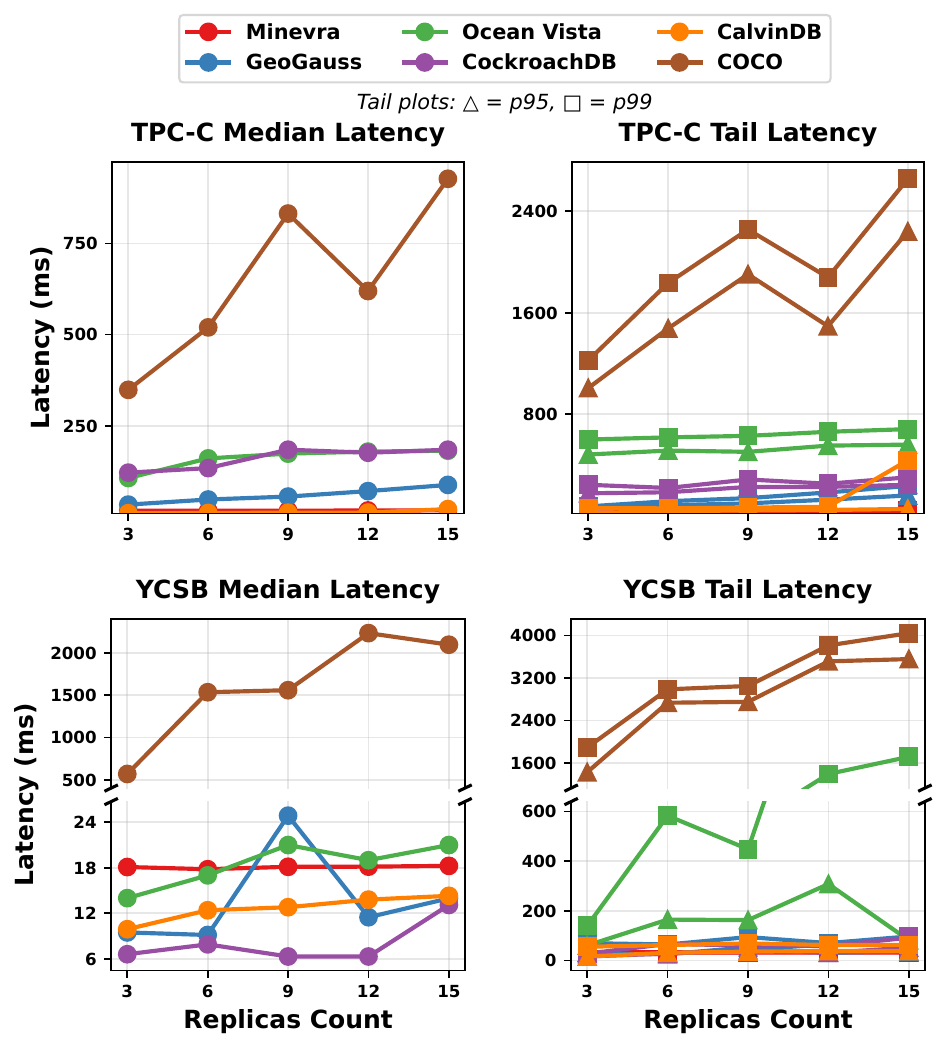}}
    \caption{Scalability in Transaction Commit Latency}
    \label{fig:ltscale}
\end{figure}

\paragraph{Latency} The latencies of individual transactions are shown in Figure~\ref{fig:ltscale}. As expected, \scmacro{}'s latency only increases slightly from $17ms$ to $18ms$ in TPC-C and from $18ms$ to $18.2ms$ in YCSB-A when scaling from 3 to 15 replicas. Most other baselines experience a steady latency except for COCO. Regarding Ocean Vista, we observe erratic fluctuations in tail latency (p95 and p99) despite a stable median, which may stem from the unpredictability of its dissemination protocol.

We see a similar trend in both the scalability and network-latency experiments for CalvinDB and \scmacro{}. This is because both systems rely on deterministic execution. While \scmacro{}'s OCC and optimized replication mechanism offer an edge in raw performance metrics, the trend remains similar, especially under high load, because \scmacro{} behaves more like a deterministic database when contention is high due to the high-contention mode.

\subsection{Contention Analysis}
\label{sec:skew}
\begin{figure}[h]
    \centerline{\includegraphics[scale=0.45] {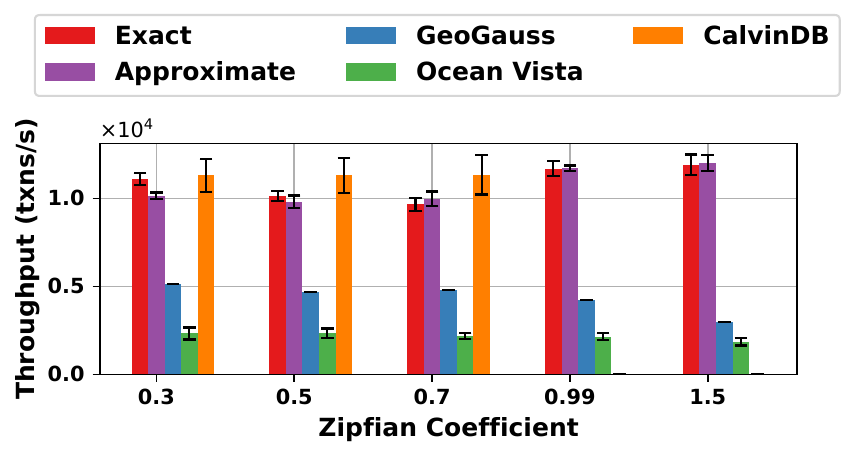}}
    \caption{Varying Contention Coefficients}
    \label{fig:skew}
\end{figure}

Figure~\ref{fig:skew} depicts the peak throughput under different Zipfian coefficients from $0.3$ to $1.7$ to change the contention of YCSB-A workload. This experiment also evaluates the performance trade-offs between the greedy approximate solver and the exact ILP solver in a LAN environment with $3$ replicas. Because CalvinDB does not support setting a Zipfian coefficient for contention, we simulate contention by setting a certain percentage of hot keys in the key space. COCO and CockroachDB~\footnote{CockroachDB only supports single operation YCSB benchmark; the 10-operation YCSB transactions are emulated using the \textit{kv} workload in other experiments, which does not support varying contentions~\cite{cdbworkload}.} are excluded from this experiment because they do not provide configurable contention control.

The exact solver slightly outperforms the approximate solver by a margin of $5\%$ at low contention, because the number of conflicts is relatively small, and the overhead of solving the MWIS problem optimally is outweighed by the benefits of minimizing re-executions. However, as the contention level increases (coefficient above $0.7$), the performance gap narrows, and the approximate solver begins to outperform the exact solver. 

At very high contention levels (coefficient above $1.2$), the overall throughput improves. This resurgence is caused by the frequent triggering of the high-contention mode, which bypasses local OCC to process transactions and proceeds straight to deterministic execution. This strategy effectively eliminates the overhead of solving MWIS on dense conflict graphs, allowing \scmacro{} to minimize scheduling latency and exhibit performance characteristics comparable to CalvinDB.

We observe about 5\% of all transactions are re-executed when the coefficient is set to 0.3, 30\% when the coefficient is 0.7, and maxes out at about 35\% for higher contention because the high-contention mode is triggered, preventing more transactions from being executed twice.

The baseline CalvinDB does not experience any performance change, as expected, since the deterministic execution model is largely unaffected by contention levels. GeoGauss and Ocean Vista both experience a performance drop as contention increases, contributing to a higher abort rate in their optimistic execution models.

\subsection{Fault Tolerance}

\begin{figure}[h]
    \centerline{\includegraphics[scale=0.5] {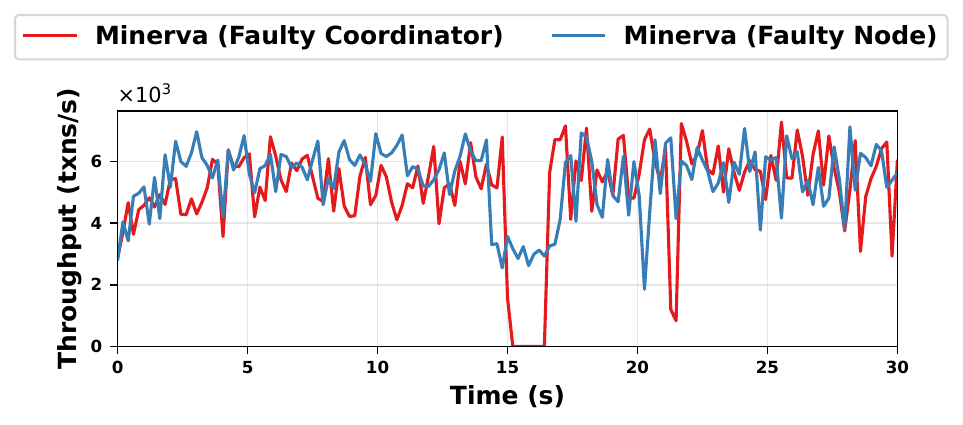}}
    \caption{Throughput under Replica Failures}
    \label{fig:failure}
\end{figure}

We evaluate \scmacro{}'s resilience to replica failures by simulating crash failures in a 5-replica cluster running the TPC-C benchmark. We introduce failures to the coordinator or a random replica at the 15-second mark during a 30-second experiment run by shutting down the server processes. The clients are not connected to the failed replica and the overall system load is set to about 50\% of the peak throughput with 400 clients across the remaining replicas to avoid saturation and better observe the recovery process.

As shown in Figure~\ref{fig:failure}, when the coordinator fails, the commit throughput is temporarily disrupted for about 1 second, which corresponds to the Raft leader election in our implementation ($400ms$ heartbeat). After the new leader is elected, the system resumes normal operation without further performance degradation. When a random replica fails, the system continues to operate with a brief drop in throughput. This is because the Raft protocol in our implementation~\cite{DotNext} waits for a timeout ($50ms$) for all replicas to acknowledge, then checks for a quorum instead of immediately committing as soon as a majority vote is received. However, once the system detects the failure, it removes the failed replica from the cluster configuration, and the throughput recovers to normal levels. These results demonstrate \scmacro{}'s ability to maintain high availability and performance in the face of replica failures, leveraging its underlying Raft-based consensus mechanism for fault tolerance.

\subsection{Performance Breakdown}
We profile \scmacro{}'s commit process with two experiment runs using the YCSB-A workload with 3 replicas, no added latency, and high-contention mode \textit{disabled}. The system is saturated with an average end-to-end transaction latency of about $500ms$ (at this point, the system throughput no longer increases with more clients). We measure a high-contention workload (about $48\%$ conflicting and stale operations) and a low-contention workload (about $2\%$ conflicting and stale operations). 

\begin{table}[h]
\fontsize{2pt}{2.5pt}\selectfont
\centering
\resizebox{\columnwidth}{!}{%
\begin{tabular}{cc|c|c}
\multicolumn{2}{c|}{}                                                                     & Low Contention & High Contention \\ \hline
\multicolumn{2}{c|}{Total Epoch Commit Time}                                              & 32.46          & 32.5            \\ \hline
\multicolumn{1}{c|}{\multirow{5}{*}{\begin{tabular}[c]{@{}c@{}}Conflict \\ Handling\end{tabular}}}  & Constructing Transaction Chains & 0.94           & 0.62            \\ \cline{2-4} 
\multicolumn{1}{c|}{}                                   & Check For Stale Transactions     & 6.84           & 3.28            \\ \cline{2-4} 
\multicolumn{1}{c|}{}                                   & Conflict Graph Construction     & 5.62           & 3.13            \\ \cline{2-4} 
\multicolumn{1}{c|}{}                                   & MWIS Solving                    & 2.27           & 2.16            \\ \cline{2-4} 
\multicolumn{1}{c|}{}                                   & Total Time                      & 17.34          & 10.09           \\ \hline
\multicolumn{2}{c|}{Apply OCC Write Sets}                                                    & 10.56          & 4.18            \\ \hline
\multicolumn{2}{c|}{Re-Execution}                                                         & 0.76           & 17.53           \\ \hline
\end{tabular}%
}
\caption{Breakdown of Average Commit Operations Time (ms)}
\label{tab:breakdown}
\end{table}

Table~\ref{tab:breakdown} presents the average operation-time breakdown for the serial commit process under both low-contention and high-contention scenarios. In both scenarios, the total commit time remains relatively consistent at approximately $32$ ms. However, the distribution of time across different operations varies significantly between the two scenarios.

In the low-contention scenario, the majority of the commit time is consumed by applying the OCC results (i.e., moving the write-sets to storage). This is expected, as most transactions pass the validation phase and are committed directly. Conversely, in the high-contention scenario, the dominant factor in commit time is the deterministic re-execution of transactions. This reflects the high rate of conflicts and subsequent aborts, leading to a substantial portion of the workload being re-executed.

We also observe that time-consuming operations consistently involve accessing the storage layer, even in our in-memory configuration, which only requires memory copying. For instance, checking for stale transactions requires reading the latest database snapshot version for every key in the read-set, and conflict graph construction involves building complex in-memory data structures based on all read and write keys. Surprisingly, the MWIS solving phase consumes a relatively steady, small amount of time across both scenarios, even when using the exact solver. This shows that the overhead is actually memory-bound because of data movement rather than computation-bound by our algorithms.

\subsection{Epoch Time}

\begin{figure}[h]
    \centerline{\includegraphics[scale=0.5] {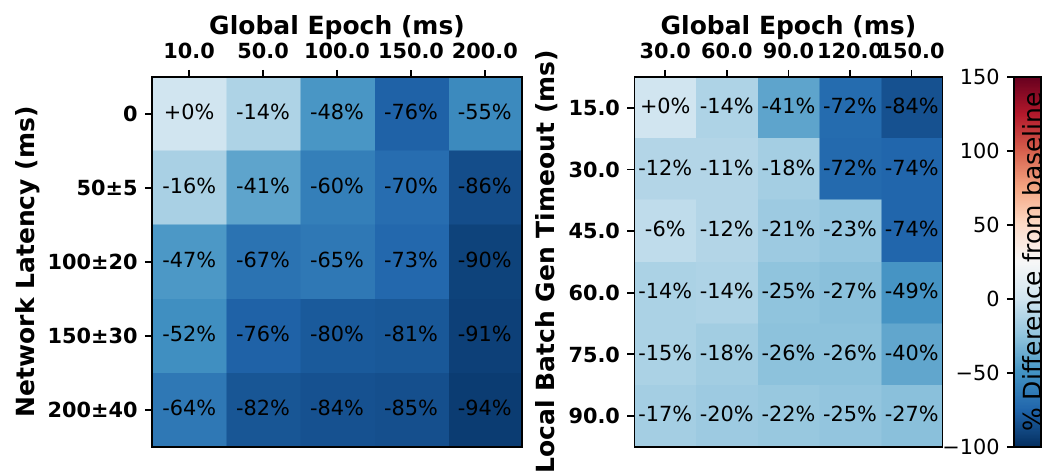}}
    \caption{Effect of Epoch Times and Batch Generation Timeouts with Varying Network Latencies}
    \label{fig:epoch}

\end{figure}

Figure~\ref{fig:epoch} illustrates the impact of varying epoch times on system throughput under different network latencies and batch generation times for the TPC-C workload with 3 replicas. The batch generation timeout is the maximum time a replica waits before broadcasting a batch, even if the batch size threshold has not been met.

Although we expect optimal performance to be gained by matching the epoch time to the network latency, the results indicate that shorter epoch times generally yield better throughput, regardless of network latency. This is because the epoch time only determines the generation frequency of the global commit proposals, while the actual commit is not synchronized to the epoch creation time. Thus, shorter epoch times allow for more frequent commit opportunities, leading to improved throughput. Similarly, a shorter batch generation time also leads to better performance, as it reduces the waiting time for transactions to be included in a batch, thereby improving overall system throughput.

\subsection{Observations and Discussion}
\label{sec:ob}
We observe that \scmacro{} is bottlenecked by the serialized commit process. Even though the commit process requires no coordination, and is asynchronous to execution and replication, it is still on the critical path for the \textit{end-to-end} transaction life cycle and thus establishes a lower bound for latency. This is essentially a producer-consumer problem where the system is limited by the rate at which the serialized commit handler can consume the optimistically produced transaction results. Therefore, regardless of how fast the local OCC execution is, the overall throughput cannot exceed the commit throughput.

We also observe a cascading effect: a slower commit results in more stale transactions because optimistic execution cannot keep up with the latest database snapshot. This, in turn, increases the number of re-executions, which further slows down the commit process. Thus, optimizing the commit process is an important future direction.

A potential improvement is to support weaker isolation and consistency levels, such as in RedBlue Consistency~\cite{li_making_2012} or Janus CRDT~\cite{maoMakingCRDTsNot2024a}. With batch-committed transactions, some conflicting operations can be resolved through commutative operations, such as those provided by Conflict-Free Replicated Data Types (CRDTs)~\cite{shapiroConflictFreeReplicatedData2011}, or by using Last-Writer-Wins (LWW) conflict resolution, rather than resorting to re-execution. Weakly consistent writes also enable a fast-path commit, where transactions can be committed as soon as they gain a PoA because the transaction is guaranteed to be persisted with the PoA.

We can also adapt the approach of Hackwrench~\cite{dongFineGrainedReExecutionEfficient2023}, which uses a predefined dependency graph to reduce the overhead in conflict detection, or utilize Aria's~\cite{luAriaFastPractical2020} transaction reorder mechanism. The trade-off for these approaches, however, is the increased complexity in transaction definition and the computational overhead required to build the dependency graph.

Another drawback of replicating both input and results is the increased bandwidth consumption. This raises cost concerns and could potentially degrade performance if the network connection bandwidth is limited. Compression techniques could be employed to mitigate this issue.

Finally, like Calvin and other deterministic databases, \scmacro{} requires transactions to be predefined with known read and write sets, and thus only supports one-shot transactions for deterministic re-execution. However, we can utilize the OCC executor for interactive transactions by using only the OCC execution path with the same replication and conflict detection mechanism, albeit with the possibility that transactions may be aborted.

\section{Related Work}
\label{sec:related}

\paragraph{Deterministic Concurrency Control}

Calvin~\cite{thomsonCalvinFastDistributed2012} first introduces deterministic concurrency control. Calvin employs a distributed sequencer layer where each database shard is assigned a sequencer to stamp arriving transactions with a global sequence number. Sequencers collect transactions into batches independently, and the commit phase is based on a network-synchronized clock. Finally, transactions in the batch are executed deterministically according to their sequence number. To maximize throughput, concurrent data access during execution is managed by a deterministic locking protocol. While Calvin is not built for multi-leader transaction execution, every replica can still accept client requests and route them to the leader replica for replication via a Paxos-based consensus protocol (implemented with ZooKeeper) to ensure durability before execution.

Aria~\cite{luAriaFastPractical2020} proposes a deterministic execution algorithm that utilizes snapshots for optimistic execution on batches of transactions. This design facilitates high concurrency without requiring a priori knowledge of read/write sets, in contrast to systems relying on deterministic locking or dependency analysis. Aria also has an execution phase and a validation phase, which identifies and aborts conflicting transactions. However, aborted transactions must be rescheduled in a subsequent batch. In high-contention scenarios, this leads to performance degradation and necessitates a synchronous re-execution fallback, a strategy similar to our conflict resolution mechanism. Replication is not part of Aria's protocol as in \scmacro{}.

HDCC~\cite{HongHDCC2025} is a recent work that also combines OCC and deterministic concurrency control by rescheduling aborted transactions with a deterministic scheduler through known read/write sets. Like Aria, HDCC does not consider replication either. Caracal~\cite{qin_caracal_2021} also introduces a similar epoch-batch-based deterministic concurrency control protocol. 

The main distinction between \scmacro{} and these deterministic systems is that \scmacro{} replicates both transaction inputs and the speculative execution results (write-sets). This dual-replication approach not only enables faster commit time for non-conflicting transactions, but also serves as a method for obtaining write-sets in deterministic execution, which otherwise requires pre-declaration or reconnaissance operations.

\paragraph{Epoch-based Replication and Commitment}
GeoGauss~\cite{zhougeogauss2023} is a recent geo-replicated OLTP database built on openGauss~\cite{openGauss} that supports multi-leader replication. Like \scmacro{}, GeoGauss is OCC-based, and replicas exchange their write-sets with all peers at the end of each epoch. These sets are then merged into a final commit set using deterministic, commutative merge rules derived from delta-state CRDTs~\cite{almeida_delta_2018}. GeoGauss resolves conflicts by deterministically aborting one of the contending transactions based on sequence numbers. \scmacro{}'s MWIS-based conflict resolution optimizes this process by reducing the number of aborted/re-executed transactions. 

Additionally, GeoGauss's epoch boundary functions as a global synchronization barrier, requiring every replica to await updates from all other replicas. While it utilizes Raft for configuration management to monitor replicas, such synchronization may cause uncertain blocking delays upon failures.

COCO~\cite{luEpochbasedCommitReplication2021} targets geo-distributed environments but adopts an asynchronous primary-backup model rather than a multi-leader architecture. It relies on sharding to achieve horizontal scalability. In the execution phase, transactions are executed optimistically and read from the nearest replica, generating write-sets that are asynchronously propagated to all replicas. At the boundary of each epoch, a coordinator initiates an optimized 2PC-based protocol to commit the accumulated batch. For validation, COCO relies on physical (through a coordinator) or logical timestamps to allow replicas to independently verify the serializability of transactions. 

Hackwrench~\cite{dongFineGrainedReExecutionEfficient2023} uses a two-tiered, batch-based commitment protocol. For conflict resolution, Hackwrench re-executes aborted transactions based on a predefined dataflow/dependency graph, contrasting with \scmacro{}'s runtime dependency analysis. Hackwrench's replication is achieved through a leaderless quorum-based write to storage nodes, effectively decoupling the transaction execution plane from storage. Neither COCO nor Hackwrench consider multi-leader replication. 

\paragraph{Asynchronous Concurrency Control and Replication}
Ocean Vista (OV) enables multi-leader geo-replication through a unified visibility control protocol~\cite{fanOceanVistaGossipbased2019}. It decouples replication from commitment with a two-phase approach: an initial S-phase uses lightweight consensus to replicate transaction "functors" (logic and inputs) across a quorum, similar to deterministic databases' transaction input replication. Final execution and visibility (akin to actually committing the transactions so clients can see the result) decisions occur in an asynchronous E-phase, where nodes independently apply transactions only after a global gossip protocol establishes a stability watermark, ensuring all prior transactions are replicated. The order of transactions is determined by timestamps synchronized across replicas using NTP, so OV optimizes for visibility latency but is sensitive to clock skew; it still requires consensus to ensure that transactions are replicated properly. 

Epidemic algorithms~\cite{holliday2003epidemic} use only asynchronous gossip protocols to propagate updates and maintain consistency without a centralized coordinator to maximize availability. They rely on causal ordering and local stability verifications to ensure serializability. However, this approach may experience unpredictable convergence times due to the probabilistic nature of message dissemination. Furthermore, ensuring strict consistency requires nodes to await updates from all peers to detect concurrent conflicts, meaning a single slow node or network failure can indefinitely stall progression. 

Mako~\cite{shenMakoSpeculativeDistributed2025} uses a similar speculative execution model and asynchronous replication. It relies on primary-backup replication, but utilizes a deterministic replay model on the replicas to reach consistency rather than replicating the results. The 2PC is optimized by committing speculatively as long as the leader shards agree to commit.

\paragraph{Other Optimized Protocols}
TAPIR~\cite{zhang_tapir_2018} abandons linearizable replication in favor of an inconsistent replication protocol with an OCC-style validation, though it still relies on client execution and client-side 2PC for atomic commitment. The C\&C framework~\cite{maiyyaCNC2019} combines atomic commitment and replication by piggybacking commit decisions onto consensus messages, yet it still requires a separate distributed concurrency control mechanism. Similarly, Janus~\cite{mu_consolidating_2016} combines concurrency control and replication by tracking dependency graphs to detect conflicts, invoking consensus only when cycles exist. 

Compared to \scmacro{}, these systems target different communication models or fault-tolerance guarantees, or a specific layer, rather than a unified protocol for the multi-leader architecture. 

\section{Conclusions}
In this paper, we present \scmacro{}, a novel distributed concurrency control protocol for transactional geo-replicated distributed databases. \scmacro{} leverages the benefits of optimistic execution for scalable parallelism across nodes, combined with a novel, highly concurrent, and fault-tolerant replication protocol. Deterministic re-execution is used to ensure that transaction aborts are minimized without requiring additional coordination among replicas. Our experimental results show that \scmacro{} outperforms state-of-the-art solutions, particularly in high-latency and high-replication-factor scenarios.

\begin{acks}
This work was in part supported by NSERC and ORF. Generative AI tools such as ChatGPT were used to assist in the writing and editing of this paper.
\end{acks}

\newpage
\bibliographystyle{ACM-Reference-Format}
\bibliography{ref}

\received{October 2025}
\received[revised]{January 2026}
\received[accepted]{February 2026}


\end{document}